\documentclass[a4paper, fleqn, 11pt]{article}

\usepackage{graphicx}
\usepackage[fleqn]{amsmath}
\usepackage{float}
\usepackage{amssymb}
\usepackage{graphicx}
\usepackage{textcomp}
\usepackage{titlesec}
\usepackage{lastpage}
\usepackage{listings}
\usepackage{setspace}
\usepackage{amsthm}
\usepackage{stfloats}
\usepackage{amsmath}

\usepackage{enumitem}
\setenumerate[1]{itemsep=0pt,partopsep=0pt,parsep=\parskip,topsep=5pt}
\setitemize[1]{itemsep=0pt,partopsep=0pt,parsep=\parskip,topsep=5pt}

\usepackage[margin=1in]{geometry}

\theoremstyle{plain}
\newtheorem{Thm}{Theorem}
\newtheorem{Lem}{Lemma}

\theoremstyle{definition}
\newtheorem{Def}{Definition}
\newtheorem{Mech}{Mechanism}
\newtheorem{Exmp}{Example}

\title{Strategyproof Mechanism for Two Heterogeneous Facilities with Constant Approximation Ratio}

\author{
	Minming Li \thanks{City University of Hong Kong, Hong Kong, China. minming.li@cityu.edu.hk}
	\and
	Pinyan Lu \thanks{ITCS, Shanghai University of Finance and Economics, Shanghai, China. lu.pinyan@mail.shufe.edu.cn}
	\and
	Yuhao Yao \thanks{University of Chinese Academy of Sciences, Beijing, China. yaoyuhao15@mails.ucas.ac.cn}
	\and
	Jialin Zhang \thanks{Institute of Computing Technology, Chinese Academy of Sciences, Beijing, China. University of Chinese Academy of Sciences, Beijing, China. zhangjialin@ict.ac.cn}
}

\date{}

\allowdisplaybreaks[4]

\begin{document}

\newpagestyle{fancy}{
    \sethead{}{}{}
    \setfoot{}{\thepage}{}
}
\pagestyle{fancy}

\maketitle

\setlength{\baselineskip}{13pt}
\setlength{\parskip}{4pt}

\begin{abstract}
In this paper, we study the two-facility location game on a line with optional preference where the acceptable set of facilities for each agent could be different and an agent's cost is his
 distance to the closest facility within his acceptable set. The objective is to minimize the total cost of all agents while
  achieving strategyproofness. We design a deterministic strategyproof mechanism for the problem with approximation ratio of $2.75$, 
  improving upon the earlier best ratio of $n/2+1$ \cite{yuan2016ecai}. 
\end{abstract} 


\section{Introduction}

Facility location games have been studied extensively in recent years. In the general setting, an authority adopts a mechanism to decide where to place the  facilities in the existence of participating agents with their reported locations or preferences. The authority aims to minimize the social cost on the condition that agents will not gain by misreporting their information. This property is called strategyproofness.
The cost of an agent is defined as his distance to the facility. In the most classic setting for single facility on a line and private location information, the median mechanism is strategyproof and minimizes the social cost \cite{moulin1980strategy}. Procaccia and Tennenholtz introduced the setting for two-facility location games and initiated the approximate mechanism design for facility location games \cite{procaccia2009}. In this two-facility model, an agent's cost is his distance to the closer facility.  Unlike the single facility setting, the optimality no longer coincides with strategyproofness in this model. They proved constant lower bound for the approximation ratio and provide a strategyproof mechanism with linear $O(n)$ approximation ratio. This left a huge gap for the tight ratio. Lu et al. obtained a randomized mechanism which achieves a constant approximation ratio of $4$ and proved a linear lower bound $\Omega(n)$ for all deterministic strategyproof mechanisms \cite{LuSWZ10}.

In their model, two facilities are identical and an agent is interested in the one closer to him. For two heterogeneous facilities, some agent may only accept one of them and the cost is his distance to that particular one. Of course, if every agent is only interested in one of them, it becomes two separate single facility games and not interesting. The interesting setting is that some agents are interested in both facilities while others are interested in only one of them, and these preferences are agents' private information. In general, an agent's cost is his distance to the closest facility in his acceptable set. This heterogeneous facilities model was proposed in \cite{yuan2016ecai}. To study the effect of private preference, they assumed that the location information is public in this new model. This is also a more realistic model since in many real facility location applications, agents' locations can be publicly verified.

In their paper, they proved a lower bound of $2$ and provided a strategyproof mechanism with linear $O(n)$ approximation ratio. This again left a huge gap for the tight ratio. It is also not clear whether there is a similar large gap between randomized and deterministic mechanisms as the homogeneous two-facility setting.

Besides its application relevance, the facility location game is also a very important playground for the theory of strategyproof mechanisms without money. Due to the remarkable Arrow's impossibility theorem \cite{arrow1950difficulty} and Gibbard-Satterthwaite's impossibility theorem \cite{gibbard1973jes}\cite{satterthwaite1975strategy}, the space of strategyproof mechanisms is very limited. Gibbard-Satterthwaite's impossibility theorem states that any strategyproof social choice function onto a set of alternatives with at least three options is a dictatorship in the general preference domain. Basically, there are only two types of strategyproof mechanisms: dictatorship or majority voting of two options. The median mechanism is neither of them. The existence of such non trivial strategyproof mechanisms is due to single-peaked preference structure of the problem. Due to this, single-peaked preference is an important direction in the study of strategyproof mechanisms without money. Studying other variants of facility location games may lead to the discovery of other interesting preference structures and non trivial strategyproof mechanisms. In most extended version of facility location games, the proposed strategyproof mechanisms are some versions (maybe randomized) of dictatorship, a majority voting of two boundary options or things like that.
The proportional mechanism from Lu et al. \cite{LuSWZ10} is one notably nontrivial example. But this one only satisfies a weaker version of strategyproofness, namely strategyproofness in expectation. This might be the reason why there is a huge gap between randomized and deterministic mechanisms. It would be more interesting if there are other nontrivial deterministic mechanisms.


\subsection{Our Contribution}
In this paper, we design a novel strategyproof mechanism which achieves an approximation ratio of $2.75$. This asymptotically resolves the problem of the two heterogeneous facilities game. Our mechanism is very simple. We first ignore agents' preferences and assume that both facilities are acceptable for each agent. We then, purely based on the public location information, compute the optimal locations of two facilities, which we call $s_{\ell}$ and $s_r$. On a line, the optimal location can be computed in polynomial time. Then among four possible locations of the two facilities $(s_{\ell}, s_{\ell})$, $(s_{\ell}, s_r)$, $(s_r, s_{\ell})$, $(s_r, s_r)$, we choose the one with the minimum social cost based on agents' reported preferences.

Our mechanism is deterministic and very simple but at the same time highly non-trivial in the sense that its strategyproofness requires an interesting argument.
It does not belong to the two exception families of Gibbard-Satterthwaite's impossibility theorem.
The reason for the median mechanism being strategyproof is the single-peaked preference. It is an interesting open direction to figure out what kind of preference structure makes our mechanism strategyproof. This may lead to a generalized version of our mechanism, which works for a family of problems.

We also studied a natural extension of our mechanism for the facility location game with three or more facilities. We provide an example to show that the mechanism is no longer strategyproof. This is similar to what happens in the proportional mechanism \cite{LuSWZ10}, whose three-facility version is no longer strategyproof. It is an interesting open question to extend our result to the facility location game with three or more facilities. We note that the same problem for homogeneous facilities remains open after almost a decade.



\subsection{Related works}

There are two major perspectives to approach the facility location game, namely, characterizing strategyproof mechanisms and designing strategyproof mechanisms. The classic agent preference for facility location games on a line is a special case of single-peaked preference for which Moulin \cite{moulin1980strategy} characterized all the anonymous, strategyproof and efficient mechanisms and at the same time showed that median mechanism is strategyproof for minimizing social cost. This is the pioneering work in the field. Later on, more characterizations are studied on networks \cite{schummer2002strategy}, for two facilities \cite{fotakis2014power} and for double-peaked preference \cite{filosratsikas2017}. There are also characterizations for other variants of facility location games.

When optimality no longer coincides with strategyproofness, Procaccia and Tennenholtz initiated the approximate mechanism design for facility location games \cite{procaccia2009}. Since then, facility location games are enriched from many different perspectives.

In terms of preference of agents for a single facility, Cheng et al. \cite{cheng2011cocoa} initiated obnoxious facility location games where agents want to be far away from the facility, followed by \cite{ibara2012cocoa}\cite{oomine2016itis}\cite{mei2018joco}. \cite{zou2015aamas} and \cite{feigenbaum2015aaai} studied dual preference where agents have different preferences towards the single facility. Filos-Ratsikas et al. \cite{filosratsikas2017} studied double-peaked preference.

In terms of the agents' cost function with respect to the distance to the facility, besides linear functions, there are also other functions studied. For example, threshold based linear function \cite{lilithreshold}, agent dependent linear function (happiness) \cite{mei2016aamas} and concave function \cite{concave}.

In terms of social objective, besides the mostly studied social cost/utility and maximum cost/utility, there are also other objectives studied like sum of square of distances\cite{feldman2013ec}\cite{mei2018joco}, difference of maximum distance and minimum distance (envy) \cite{cai2016ijcai}.

The studies on heterogeneous facilities was initiated by Serafino and Ventre\cite{serafino2014ecai}\cite{serafino2015aaai}, where the cost of an agent who likes both facilities is the sum of the distances to both facilities. When the cost is only affected by the closer facility or the farther facility, it is called optional preference and is studied by \cite{yuan2016ecai}. Recently fractional preference is studied by \cite{fong2018aaai}, where the cost is a weighted sum of the two distances. There are also works setting a distance constraint between the two facilities, either maximum distance \cite{zou2015aamas}\cite{chen2018mechanism} or minimum distance \cite{AAMAS2019c}.

For other metric spaces, Alon et al. \cite{alon2010mor} and Cheng et al. \cite{cheng2013tcs} studied networks. Lu et al. \cite{LuSWZ10} studied Euclidean space where a 4-approximated randomized strategyproof mechanism is proposed.

There are also recent works on capacitated facility location games \cite{aziz2018capacity}, dynamic facility location games \cite{IJCAI2018}\cite{wada2018facility}, facility location games with externality \cite{AAMAS2019a} and facility location games with dual roles \cite{AAMAS2019b} where agents also share the role of facilities and payment is introduced to deal with a non-standard utility model.

More works could be found in \cite{nehama2019manipulations}\cite{anastasiadis2018aamas}\cite{dokow2012mechanism}\cite{lu2009wine}\cite{sonoda2016aaai}\cite{todo2011aamas}\cite{zhang2014joco}\cite{tamir1991obnoxious}\cite{thang2010group}\cite{cheng2013obnoxious}.

\section{Preliminaries}


There are $n$ agents on a line and the government is to build two facilities (named as $F_1$ and $F_2$) for agents. Agent $i$ is located at $x_i$ and has a preference $p_i \in \{\{F_1\}, \{F_2\}, \{F_1,F_2\}\}$. (Without loss of generality, we assume $x_1 \le x_2 \le ... \le x_n$.) Here, we call \textbf{x} = $\{x_1,x_2,...,x_n\}$ as location profile and \textbf{p} = $\{p_1,p_2,...,p_n\}$ as preference profile.


In this model, we define the distance between two points $s$ and $t$ as $dist(s,t)=|s-t|$.

Given that $F_1$ is built at $y_1$ and $F_2$ is built at $y_2$, we define cost of the agent $i$ as $cost_{\textbf{x},\textbf{p},i}(y_1,y_2) = \min_{k: F_k \in p_i} dist(x_i,y_k)$. Then, the social cost is defined as $COST_{\textbf{x},\textbf{p}}(y_1,y_2)$ = $\sum_{i=1}^{n} cost_{\textbf{x},\textbf{p},i}(y_1,y_2)$.

In this model, the government uses a \textit{mechanism} to decide the places to build the facilities. More specifically, a mechanism takes a location profile \textbf{x} and a preference profile \textbf{p} as input and outputs facilities' locations.

Given that an agent might misreport his preference, which will lead to a different output by the mechanism and might lower his own cost, the government needs to design a \textit{strategyproof} mechanism to avoid this situation.

\begin{Def}
	(strategyproof)
	
	A mechanism $M$ is strategyproof if $M$ has the following property:

	For any location and preference profile \textbf{x} and \textbf{p}, any agent $i$ does not have incentive to misreport his preference. That is to say if the agent $i$ cheats by misreporting his preference as $p'_i$ and the preference profile changes to \textbf{p'} = $\{p_1 ,p_2, ..., p_{i-1}, p'_i, p_{i+1}, ..., p_n\}$, then there must be $cost_{\textbf{x},\textbf{p},i}(M(\textbf{x},\textbf{p})) \le cost_{\textbf{x},\textbf{p},i}(M(\textbf{x},\textbf{p'}))$.
\end{Def}

Given that no strategyproof mechanism achieves the minimum social cost \cite{yuan2016ecai}, we use \textit{approximation ratio} to evaluate the performance of mechanisms.

\begin{Def}
	(approximation ratio)
	
	A mechanism $M$ is $\alpha$-approximated if for any location and preference profile \textbf{x} and \textbf{p}, $COST_{\textbf{x},\textbf{p}}(M(\textbf{x},\textbf{p})) \le \alpha \cdot \min_{y_1,y_2} COST_{\textbf{x},\textbf{p}}(y_1,y_2)$.
\end{Def} 
\section{Our Strategyproof Mechanism}
In this section, we introduce our mechanism and prove its strategyproofness. In Section 4, we will analyze its approximation ratio.

\subsection{The mechanism}

\begin{Mech}
	Define preference profile \textbf{q} = $\{\{F_1, F_2\}, \{F_1, F_2\}, ..., \{F_1, F_2\}\}$, i.e. every agent prefers both facilities in the preference profile \textbf{q}. Let location pair $(s_{\ell}, s_r)$ satisfy $(s_{\ell}, s_r) = \mathop{\arg\min} _{y_1, y_2 \in \{x_1, x_2, ..., x_n\}} COST_{\textbf{x},\textbf{q}}(y_1, y_2)$. If $(y_1, y_2)$ which minimizes $COST_{\textbf{x},\textbf{q}}(y_1, y_2)$ is not unique, we choose the pair with minimum $y_1$ first and then minimum $y_2$. Finally, we output the pair $(f_1, f_2) = \mathop{\arg\min} _{y_1, y_2 \in \{s_{\ell}, s_r\}} COST_{\textbf{x},\textbf{p}}(y_1, y_2)$. If there are multiple solutions, we choose our output in this order: $(s_{\ell}, s_{\ell})$, $(s_{\ell}, s_r)$, $(s_r, s_{\ell})$, $(s_r, s_r)$.
\end{Mech}

In the rest of this paper, $s_{\ell}$ and $s_r$ are specified to be the $s_{\ell}$ and $s_r$ calculated in Mechanism 1.

\subsection{Strategyproofness}

\begin{Thm}
	Mechanism 1 is strategyproof.
\end{Thm}

\begin{proof}
It is easy to see that $s_{\ell}$ and $s_r$ do not rely on agents' preferences and the two facilities' locations will always be one among $(s_{\ell}, s_{\ell})$, $(s_{\ell}, s_r)$, $(s_r, s_{\ell})$ and $(s_r, s_r)$.

We first prove that the agent whose preference is $\{F_1, F_2\}$ will not cheat. Suppose $p_i = \{F_1, F_2\}$. Without loss of generality, we assume that $s_r$ is closer to agent $i$ than $s_{\ell}$. Therefore, the only case that agent $i$ has incentive to cheat is that Mechanism 1's output is $(s_{\ell}, s_{\ell})$. Let \textbf{p} be the preference profile where no agent cheats and \textbf{p'} be the preference profile where agent $i$ misreports his preference as $\{F_1\}$. Table 1 shows the difference of social cost for \textbf{p'} and \textbf{p}.

\begin{table}[h]
\begin{center}
\begin{tabular}{|c|c|}
\hline
$(y_1, y_2)$ & $COST_{\textbf{x},\textbf{p'}}(y_1,y_2) - COST_{\textbf{x},\textbf{p}}(y_1,y_2)$ \\
\hline
$(s_{\ell}, s_{\ell})$ & 0 \\
\hline
$(s_{\ell}, s_r)$ & $dist(x_i, s_{\ell}) - dist(x_i, s_r)\geq 0$ \\
\hline
$(s_r, s_{\ell})$ & 0 \\
\hline
$(s_r, s_r)$ & 0 \\
\hline
\end{tabular}
\caption{the difference of social cost for \textbf{p'} and \textbf{p}}
\end{center}
\end{table}

Since Mechanism 1's output is $(s_{\ell}, s_{\ell})$ for preference profile \textbf{p}, $COST_{\textbf{x},\textbf{p}}(s_{\ell}, s_{\ell})$ is the minimum among those four social costs for \textbf{p}. By Table 1, it is easy to see that $COST_{\textbf{x},\textbf{p'}}(s_{\ell}, s_{\ell})$ is the minimum among those four social costs for \textbf{p'}. This shows that Mechanism 1's output will still be $(s_{\ell}, s_{\ell})$ even if agent $i$ misreports his preference as $\{F_1\}$. Therefore, agent $i$ has no incentive to misreport his preference as $\{F_1\}$. Using the same analysis, we can prove that agent $i$ has no incentive to misreport his preference as $\{F_2\}$. Hence, agent $i$ has no incentive to cheat.

We then prove that the agent whose preference is $\{F_1\}$ or $\{F_2\}$ will not cheat. Suppose that $p_i = \{F_1\}$. (The proof for the agent with preference $\{F_2\}$ is the same.) Without loss of generality, we assume that $s_r$ is closer to agent $i$ than $s_{\ell}$. Hence, agent $i$ might cheat only when $F_1$ is built at $s_{\ell}$. At this time, $F_2$ may be built at $s_{\ell}$ or $s_r$. Let \textbf{p} be the preference profile where no agent cheats, \textbf{p'} be the preference profile where agent $i$ misreports his preference as $\{F_2\}$ and \textbf{p''} be the preference profile where agent $i$ misreports his preference as $\{F_1, F_2\}$. Table 2 shows the difference of social cost for \textbf{p'} and \textbf{p} and for \textbf{p''} and \textbf{p}.

\begin{table}[h]
\begin{center}
\begin{tabular}{|c|c|c|}
\hline
$(y_1, y_2)$ & $COST_{\textbf{x},\textbf{p'}}(y_1,y_2) - COST_{\textbf{x},\textbf{p}}(y_1,y_2)$ & $COST_{\textbf{x},\textbf{p''}}(y_1,y_2) - COST_{\textbf{x},\textbf{p}}(y_1,y_2)$\\
\hline
$(s_{\ell}, s_{\ell})$ & 0 & 0\\
\hline
$(s_{\ell}, s_r)$ & $dist(x_i, s_r) - dist(x_i, s_{\ell})\leq 0$ & $dist(x_i, s_r) - dist(x_i, s_{\ell})\leq 0$\\
\hline
$(s_r, s_{\ell})$ & $dist(x_i, s_{\ell}) - dist(x_i, s_r)\geq 0$ & 0 \\
\hline
$(s_r, s_r)$ & 0 & 0 \\
\hline
\end{tabular}
\caption{the difference of social cost for \textbf{p'} and \textbf{p} and for \textbf{p''} and \textbf{p}}
\end{center}
\end{table}

If $F_2$ is built at $s_{\ell}$ for \textbf{p}, then $COST_{\textbf{x},\textbf{p}}(s_{\ell}, s_{\ell}) \le \min \{COST_{\textbf{x},\textbf{p}}(s_r, s_{\ell}), COST_{\textbf{x},\textbf{p}}(s_r, s_r)\}$. By Table 2, it is easy to see that $COST_{\textbf{x},\textbf{p'}}(s_{\ell}, s_{\ell}) \le \min \{COST_{\textbf{x},\textbf{p'}}(s_r, s_{\ell}), COST_{\textbf{x},\textbf{p'}}(s_r, s_r)\}$ and $COST_{\textbf{x},\textbf{p''}}(s_{\ell}, s_{\ell}) \le \min \{COST_{\textbf{x},\textbf{p''}}(s_r, s_{\ell}), COST_{\textbf{x},\textbf{p''}}(s_r, s_r)\}$. This means $F_1$ will not be built at $s_r$ even if agent $i$ misreports his preference as $\{F_2\}$ or $\{F_1, F_2\}$.

If $F_2$ is built at $s_r$ for \textbf{p}, then $COST_{\textbf{x},\textbf{p}}(s_{\ell}, s_r) \le \min \{COST_{\textbf{x},\textbf{p}}(s_r, s_{\ell}), COST_{\textbf{x},\textbf{p}}(s_r, s_r)\}$. By Table 2, it is easy to see that $COST_{\textbf{x},\textbf{p'}}(s_{\ell}, s_r) \le \min \{COST_{\textbf{x},\textbf{p'}}(s_r, s_{\ell}), COST_{\textbf{x},\textbf{p'}}(s_r, s_r)\}$ and $COST_{\textbf{x},\textbf{p''}}(s_{\ell}, s_r) \le \min \{COST_{\textbf{x},\textbf{p''}}(s_r, s_{\ell}), COST_{\textbf{x},\textbf{p''}}(s_r, s_r)\}$. Similarly, this means that agent $i$'s misreporting will not make $F_1$ closer. Therefore agent $i$ has no incentive to cheat.

\end{proof}

\section{Analysis on the approximation ratio of Mechanism 1}

In this section, we analyze the approximation ratio of Mechanism 1.

We first introduce some lemmas in Section 4.1. In Section 4.2, we show that Mechanism 1 is 3-approximated. After that, we show how to improve the approximation ratio from 3 to 2.75 in Section 4.3. Finally, in Section 4.4, we show a lower bound of the approximation ratio, which we conjecture to be the exact approximation ratio of Mechanism 1.

\subsection{Preparation}

\begin{Lem}
When analyzing the approximation ratio of Mechanism 1 (denoted as $M$), we can assume $p_i \ne \{F_1,F_2\}$ for every $i \in \{1, 2, ..., n\}$.
\end{Lem}

\begin{proof}
Suppose we have a location profile \textbf{x} and a preference profile \textbf{p} with $p_i = \{F_1,F_2\}$. Let $(opt_1, opt_2)$ be the locations such that $COST_{\textbf{x},\textbf{p}}(opt_1, opt_2) = \min_{y_1, y_2} COST_{\textbf{x},\textbf{p}}(y_1, y_2)$. Let $p'_i$ be $\{F_1\}$ if $dist(x_i, opt_1) < dist(x_i, opt_2)$. Otherwise, let $p'_i$ be $\{F_2\}$. Let preference profile \textbf{p'} = $\{p_1, p_2, ..., p_{i-1}, p'_i, p_{i+1}, ..., p_n\}$, i.e. we get \textbf{p'} by replacing $p_i$ with $p'_i$ in \textbf{p}.

We first prove $COST_{\textbf{x},\textbf{p}}(M(\textbf{x},\textbf{p})) \le COST_{\textbf{x},\textbf{p'}}(M(\textbf{x},\textbf{p'}))$. Since the location profiles in (\textbf{x},\textbf{p}) and (\textbf{x},\textbf{p'}) are the same, $(s_{\ell}, s_r)$ in $M(\textbf{x},\textbf{p})$ and $M(\textbf{x},\textbf{p'})$ are the same. Note that we have $COST_{\textbf{x},\textbf{p}}(y_1, y_2) \le COST_{\textbf{x},\textbf{p'}}(y_1, y_2)$ for any $(y_1, y_2)$, since $p'_i \subseteq p_i$ and other agents' preferences are the same. Therefore, we have $COST_{\textbf{x},\textbf{p}}(y_1, y_2) \le COST_{\textbf{x},\textbf{p'}}(y_1, y_2)$ for $y_1, y_2 \in \{s_{\ell}, s_r\}$, which implies 
\begin{align*}
COST_{\textbf{x},\textbf{p}}(M(\textbf{x},\textbf{p})) 
&= \min\limits_{y_1,y_2 \in \{s_{\ell},s_r\}} COST_{\textbf{x},\textbf{p}}(y_1,y_2) \\
&\le \min\limits_{y_1,y_2 \in \{s_{\ell},s_r\}} COST_{\textbf{x},\textbf{p'}}(y_1,y_2) \\
&= COST_{\textbf{x},\textbf{p'}}(M(\textbf{x},\textbf{p'})).
\end{align*}

We then prove $\min_{y_1,y_2} COST_{\textbf{x},\textbf{p}}(y_1,y_2) \ge \min_{y_1,y_2} COST_{\textbf{x},\textbf{p'}}(y_1,y_2)$. By the definition of $p'_i$, we have $COST_{\textbf{x},\textbf{p'}}(opt_1, opt_2) = COST_{\textbf{x},\textbf{p}}(opt_1, opt_2)$. Therefore, we have 
\begin{align*}
\min\limits_{y_1,y_2} COST_{\textbf{x},\textbf{p'}}(y_1,y_2) 
&\le COST_{\textbf{x},\textbf{p'}}(opt_1, opt_2) \\
&= COST_{\textbf{x},\textbf{p}}(opt_1, opt_2) \\
&= \min\limits_{y_1, y_2} COST_{\textbf{x},\textbf{p}}(y_1, y_2). 
\end{align*}

Then, $COST_{\textbf{x},\textbf{p'}}(M(\textbf{x},\textbf{p'})) \le \alpha \cdot \min_{y_1,y_2} COST_{\textbf{x},\textbf{p'}}(y_1,y_2)$ will imply $COST_{\textbf{x},\textbf{p}}(M(\textbf{x},\textbf{p})) \le \alpha \cdot \min_{y_1,y_2} COST_{\textbf{x},\textbf{p}}(y_1,y_2)$. By simple induction, if we want to prove that Mechanism 1 is $\alpha$-approximated for some $\alpha$, we only need to consider the situation where no agent's preference is $\{F_1,F_2\}$.

\end{proof}

\begin{Lem}
Given a set of points $S = \{x_1, x_2, ...,x_n\}$ on a line, the median of $S$, denoted as $p$, satisfies $\sum_{i=1}^{n} dist(x_i, p) = \min_{y} \sum_{i=1}^{n} dist(x_i, y)$.
\end{Lem}

The proof of Lemma 2 is straightforward.

\begin{Lem}
Let preference profile {\rm \textbf{q}} = $\{\{F_1, F_2\}, \{F_1, F_2\}, ..., \{F_1, F_2\}\}$. Then, $COST_{{\rm \textbf{x},\textbf{q}}}(s_{\ell}, s_r) \le \min_{y_1, y_2} COST_{{\rm \textbf{x},\textbf{p}}}(y_1, y_2)$ holds for any location profile {\rm \textbf{x}} and preference profile {\rm \textbf{p}}.
\end{Lem}

\begin{proof}

By \cite{yuan2016ecai}, we have $COST_{\textbf{x},\textbf{q}}(s_{\ell}, s_r) = \min_{y_1, y_2} COST_{\textbf{x},\textbf{q}}(y_1, y_2)$. Let $(opt_1, opt_2)$ satisfy $COST_{\textbf{x},\textbf{p}}(opt_1, opt_2) = \min_{y_1, y_2} COST_{\textbf{x},\textbf{p}}(y_1, y_2)$. Then we have
\begin{align*}
COST_{{\rm \textbf{x},\textbf{q}}}(s_{\ell}, s_r)
&= \min\limits_{y_1, y_2} COST_{\textbf{x},\textbf{q}}(y_1, y_2) \\
&\le COST_{\textbf{x},\textbf{q}}(opt_1, opt_2) \\
&\le COST_{\textbf{x},\textbf{p}}(opt_1, opt_2) \\
&= \min\limits_{y_1, y_2} COST_{\textbf{x},\textbf{p}}(y_1, y_2). 
\end{align*}

\end{proof}

\subsection{Analysis of the approximation ratio of 3}

We first introduce some notations to simplify the description of later analysis.

Given a location profile \textbf{x} and a preference profile \textbf{p}, by Lemma 1, we can divide all agents into two sets $\{i|p_i = \{F_1\}\}$ and $\{i|p_i = \{F_2\}\}$. Let $opt_1$ be the median of $S_1 = \{x_i|p_i = \{F_1\}\}$ and $opt_2$ be the median of $S_2 = \{x_i|p_i = \{F_2\}\}$. (Note that $S_1$ may have same locations and we do not remove them. So does $S_2$.) By Lemma 2, it is easy to see that $COST_{\textbf{x}, \textbf{p}} (opt_1, opt_2) = \min_{y_1, y_2} COST_{\textbf{x}, \textbf{p}} (y_1, y_2)$.

Recall that $s_{\ell}$ and $s_r$ are two locations calculated in Mechanism 1, where $s_{\ell} \le s_r$. We denote $s_{mid}$ as the midpoint of $s_{\ell}$ and $s_r$.

For point set $S$ and point $q$, we define $d(S,q) = \sum_{p \in S} dist(p,q)$. (This definition differs from the common one where $d(S,q)$ is defined as $\min_{p \in S} dist(p,q)$.) We define some notations as follows.
\begin{itemize}
	\item $COST_k \triangleq \min \{d(S_k, s_{\ell}), d(S_k, s_r)\}$, $k = 1,2$. \\
	$COST \triangleq COST_1 + COST_2$. 
	\item $OPT_k \triangleq d(S_k, opt_k)$, $k = 1,2$.\\
	$OPT \triangleq OPT_1 + OPT_2$. 
	\item $BEST_k \triangleq \sum_{x \in S_k} \min \{dist(x, s_{\ell}), dist(x, s_r)\}$, $k = 1,2$.\\
	$BEST \triangleq BEST_1 + BEST_2$. 
\end{itemize}

Intuitively, $COST$ is equal to $COST_{\textbf{x}, \textbf{p}}(M(\textbf{x}, \textbf{p}))$, i.e. the social cost calculated by Mechanism 1. $OPT$ is equal to $\min_{y_1, y_2} COST_{\textbf{x}, \textbf{p}}(y_1,y_2)$, i.e. the minimum social cost for \textbf{p}. $BEST$ is equal to $COST_{\textbf{x}, \textbf{q}}(s_{\ell},s_r)$, i.e. the minimum social cost for \textbf{q}. (\textbf{q} is defined in Mechanism 1 and Lemma 3.)
With these notations, Lemma 3 implies $BEST \le OPT$.

Now, we introduce our first result on the approximation ratio of Mechanism 1.

\begin{Thm}
	Mechanism 1 is 3-approximated, i.e. $COST \le 3 \cdot OPT$.
\end{Thm}

\begin{proof}
To prove the ratio, we try to prove $COST - OPT \le 2 \cdot BEST$. It is sufficient to show $COST_k - OPT_k \le 2 \cdot BEST_k$, $k = 1, 2$. 

By symmetry, we only show $COST_1 - OPT_1 \le 2 \cdot BEST_1$. We first consider the case that $opt_1$ lies on the left side of $s_{mid}$. (See Figure 1.) By definition, we have:
\begin{align*}
COST_1 - OPT_1 	&= \min \{d(S_1, s_{\ell}), d(S_1,s_r)\} - d(S_1, opt_1)\\
				&\le d(S_1, s_{\ell}) - d(S_1, opt_1).
\end{align*}

Without loss of generality, we assume that the amount of agents is even. Otherwise, we copy all agents, including both locations and preferences, and add them to the initial set of agents. It is easy to see that both $COST$ and $OPT$ double. Hence, it does not affect our analysis on the approximation ratio.

Since $opt_1$ is the median of $S_1$, half of $S_1$ lie on the left side of $opt_1$ (or at $opt_1$) and the rest lie on the right side of $opt_1$ (or at $opt_1$). We match the points between those two parts in the way brackets match, which is shown in Figure 1. Finally, those points form $|S_1|/2$ pairs $(u_j, v_j)$, where $u_j$ is in the first part (left side part) and $v_j$ is in the second part (right side part), $j \in \{1, 2, ..., |S_1|/2\}$.

\begin{figure}[h]
    \centering
    \includegraphics[width=0.80\textwidth]{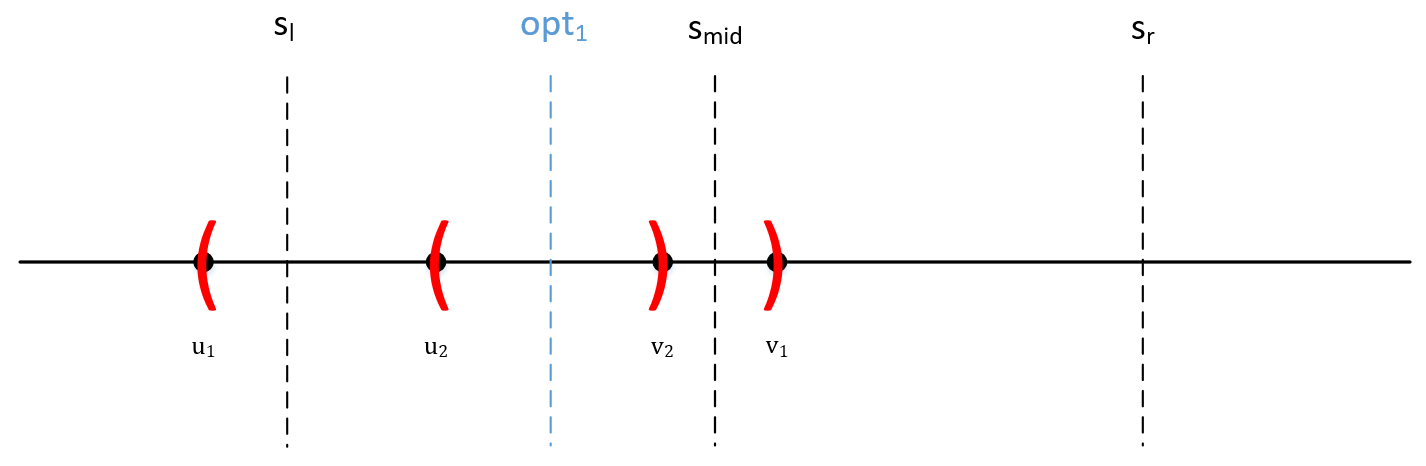}
    \caption{the case $opt_1 \le s_{mid}$}
    \label{fig1}
\end{figure}

We consider the case that $opt_1$ lies on the right side of $s_{\ell}$, then we have the following equation for any $j \in \{1, 2, ..., |S_1|/2\}$.
\begin{equation*}
 d(\{u_j,v_j\},s_{\ell}) - d(\{u_j,v_j\},opt_1) =
\begin{cases}
0, 						& \text{$u_j$ is on the left side of $s_{\ell}$}\\
2 \cdot dist(u_j,s_{\ell}), 	& \text{otherwise}
\end{cases}
\end{equation*}

Given that each $u_j$ lies on the left side of $opt_1$ and $opt_1$ lies on the left side of $s_{mid}$, we find $s_{\ell}$ is closer to $u_j$ than $s_r$. By the definition of $BEST_1$, we have
\begin{align*}
BEST_1 	&= \sum\limits_{x \in S_i} \min \{dist(x, s_{\ell}), dist(x, s_r)\}\\
		&\ge \sum\limits_{j = 1}^{|S_1|/2} \min \{dist(u_j, s_{\ell}), dist(u_j, s_r)\}\\
		&= \sum\limits_{j = 1}^{|S_1|/2} dist(u_j, s_{\ell}).
\end{align*}

Then we have 
\begin{align*}
	COST_1 - OPT_1  &\le d(S_1, s_{\ell}) - d(S_1, opt_1) \\
					&= \sum\limits_{j=1}^{|S_1|/2} (d(\{u_j,v_j\},s_{\ell}) - d(\{u_j,v_j\},opt_1))\\
					&= \sum\limits_{j: \text{$u_j$ is on the right side of $s_{\ell}$}} 2 \cdot dist(u_j, s_{\ell}) \\
					&\le \sum\limits_{j = 1}^{|S_1|/2} 2 \cdot dist(u_j, s_{\ell}) \\
					&\le 2 \cdot BEST_1.
\end{align*}

The analysis for the case that $opt_1$ lies on the left side of $s_{\ell}$ is similar. Hence, $COST_1 - OPT_1 \le 2 \cdot BEST_1$ holds for the case that $opt_1$ lies on the left side of $s_{mid}$. 

The case that $opt_1$ lies on the right side of $s_{mid}$ can be solved by using the same analysis. Hence, we have $COST_1 - OPT_1 \le 2 \cdot BEST_1$.

By Lemma 3, we have $COST - OPT \le 2 \cdot BEST \le 2 \cdot OPT$. Hence, $COST \le 3 \cdot OPT$.

\end{proof}

\subsection{Improvement of approximation ratio from 3 to 2.75}

For simplicity of later discussion, let $\Delta_k = (COST_k - OPT_k)/2$, $k = 1, 2$ and $\Delta = \Delta_1 + \Delta_2$.

In the previous section, we obtain the approximation ratio 3 by showing $\Delta_k \le BEST_k$, $k = 1, 2$. In this section, we will give a tighter upper bound of $\Delta_k$ depending on the position of $opt_k$. Our intuition comes from the observation that when $opt_1$ is on the left side of $s_{\ell}$, we have $\Delta_1 \le 1/2 \cdot BEST_1$ (Case 1 in the proof of Theorem 3). We then carefully analyze the case when $opt_1$ is between $s_{\ell}$ and $s_{mid}$. The main difficulty happens when $opt_1$ is close to $s_{\ell}$, where in this case, $\Delta_1$ can be close to $BEST_1$. Fortunately, we find that in such case, we can bound $\Delta_1$ by $OPT_1$. The following obstacle is that although $BEST \le OPT$, there is no relationship between $OPT_1$ and $BEST_1$. We handle this issue by simultaneously considering $\Delta_1$ and $\Delta_2$.

\begin{Thm}
	Mechanism 1 is 2.75-approximated, i.e. $COST \le 2.75 \cdot OPT$.
\end{Thm}

\begin{proof}

We try to bound $\Delta_1$ first. By symmetry, we assume that $opt_1$ is on the left side of $s_{mid}$. Let $s_c$ be the point satisfying $dist(s_c,s_{mid}) = c \cdot dist(s_{\ell},s_c)$. Here, $c$ is a constant parameter to be determined later.

We discuss in the following three cases, according to the position of $opt_1$.

\textbf{Case 1.} $opt_1$ is on the left side of $s_{\ell}$.
\begin{figure}[h]
    \centering
    \includegraphics[width=0.90\textwidth]{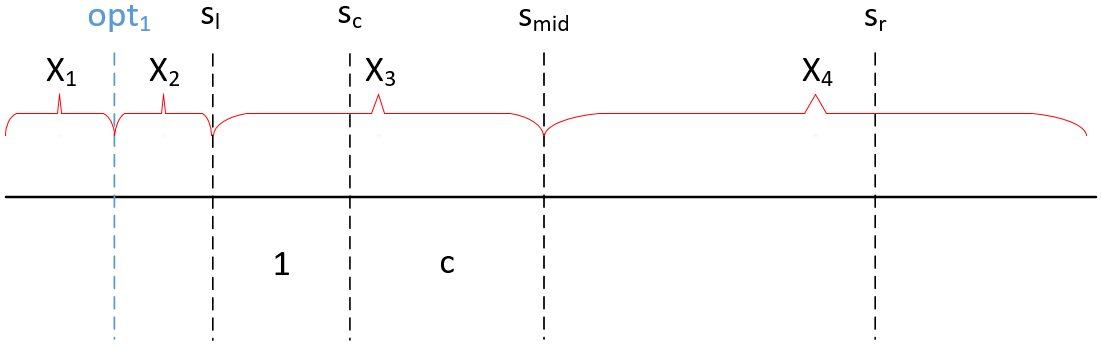}
    \caption{Case 1. $opt_1$ is on the left side of $s_{\ell}$}
    \label{Case 1}
\end{figure}

In this case, we divide $S_1$ into four parts: $X_1$, $X_2$, $X_3$ and $X_4$. (See Figure 2.)
\begin{enumerate}
	\item $X_1$ contains the points on the left of $opt_1$.
	\item $X_2$ contains the points between $opt_1$ and $s_{\ell}$.
	\item $X_3$ contains the points between $s_{\ell}$ and $s_{mid}$.
	\item $X_4$ contains the points on the right of $s_{mid}$.
\end{enumerate}\par

We allocate points on the boundary in such a way that $|X_1| = |X_2| + |X_3| + |X_4|$ (This can be achieved since $opt_1$ divides $S_1$ into two parts, left part and right part, with the same size).

Using the similar analysis as the proof of Theorem 2 and noticing that $opt_1$ is on the left side of $s_{\ell}$, we find that $\Delta_1$ is equal to the distance between $s_{\ell}$ and points in $S_1$ which lie between $opt_1$ and $s_{\ell}$, i.e. $d(X_2, s_{\ell})$. Hence, we have
\begin{align*}
	\Delta_1 	= d(X_2, s_{\ell})
				\le |X_2| \cdot dist(opt_1,s_{\ell})
				\le |X_1| \cdot dist(opt_1,s_{\ell})
				\le d(X_1, s_{\ell}).
\end{align*}

In this case,
\begin{align*}
	BEST_1 = d(X_1, s_{\ell}) + d(X_2, s_{\ell}) + d(X_3, s_{\ell}) + d(X_4, s_r) 
		   \ge d(X_1, s_{\ell}) + d(X_2, s_{\ell})
		   \ge 2 \cdot \Delta_1.
\end{align*}

This means that we have $\Delta_1 \le 1/2 \cdot BEST_1$.\\

\textbf{Case 2.} $opt_1$ is between $s_c$ and $s_{mid}$.
\begin{figure}[h]
    \centering
    \includegraphics[width=0.90\textwidth]{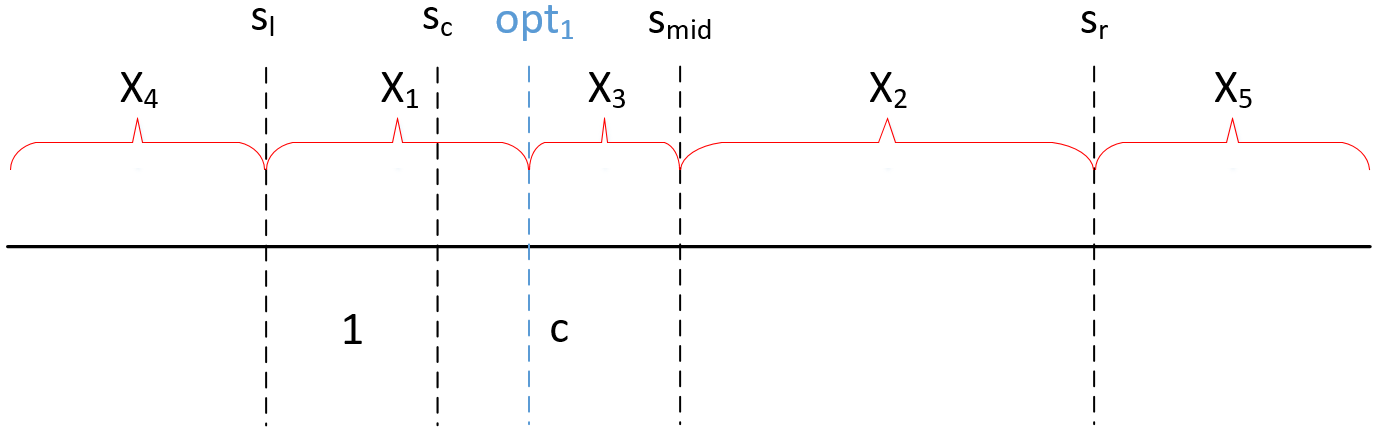}
    \caption{Case 2. $opt_1$ is between $s_c$ and $s_{mid}$}
    \label{Case 2}
\end{figure}

In this case, we divide $S_1$ into five parts: $X_1$, $X_2$, $X_3$, $X_4$ and $X_5$. (See Figure 3.)
\begin{enumerate}
	\item $X_1$ contains the points between $s_{\ell}$ and $opt_1$.
	\item $X_2$ contains the points between $s_{mid}$ and $s_r$.
	\item $X_3$ contains the points between $opt_1$ and $s_{mid}$.
	\item $X_4$ contains the points on the left of $s_{\ell}$.
	\item $X_5$ contains the points on the right of $s_r$.
\end{enumerate}

Similarly, we can assume $|X_1| + |X_4| = |X_2| + |X_3| + |X_5|$.

By the definition of $\Delta_1$, using a similar proof as Theorem 2, we have 
\begin{equation*}
	\Delta_1 = \min \{d(X_1,s_{\ell}), d(X_3, s_r) + d(X_2, s_r)\}.
\end{equation*}

In this case,
\begin{align*}
	BEST_1 &= d(X_1,s_{\ell}) + d(X_2, s_r) + d(X_3, s_{\ell}) + d(X_4, s_{\ell}) + d(X_5, s_r)\\
		   &\ge d(X_1,s_{\ell}) + d(X_2, s_r) + d(X_3, s_{\ell}).
\end{align*}

Since $dist(s_c,s_{mid}) = c \cdot dist(s_{\ell},s_c)$ and $dist(s_{\ell}, s_{mid}) = dist(s_{mid}, s_r)$, we have $dist(s_c, s_r) = (2c+1) \cdot dist(s_{\ell}, s_c)$. Given that $opt_1$ is on the right side of $s_c$ in this case, $dist(p, s_r) \le (2c+1) \cdot dist(p, s_{\ell})$ holds for every $p \in X_3$. Thus, we have $d(X_3, s_r) \le (2c+1) \cdot d(X_3, s_{\ell})$. Hence,
\begin{align*}
	d(X_3, s_r) + d(X_2, s_r) \le (2c+1) \cdot d(X_3, s_{\ell}) + d(X_2, s_r)
							  \le (2c+1) \cdot (d(X_3, s_{\ell}) + d(X_2, s_r)).
\end{align*}

If $d(X_1,s_{\ell}) \le (2c+1)/(2c+2) \cdot BEST_1$, then we have
\begin{equation*}
	\Delta_1 \le d(X_1,s_{\ell}) \le \frac{2c+1}{2c+2} BEST_1.
\end{equation*}

Otherwise, $d(X_1,s_{\ell}) > (2c+1)/(2c+2) \cdot BEST_1$. Then, we have
\begin{align*}
	\Delta_1 &\le d(X_3, s_r) + d(X_2, s_r)\\
			 &\le (2c+1) \cdot (d(X_3, s_{\ell}) + d(X_2, s_r))\\
			 &\le (2c+1) \cdot (BEST_1 - d(X_1,s_{\ell}))\\
			 &< \frac{2c+1}{2c+2} BEST_1.
\end{align*}

Therefore, in this case, we have $\Delta_1 \le (2c+1)/(2c+2) \cdot  BEST_1$.\\

\textbf{Case 3.} $opt_1$ is between $s_{\ell}$ and $s_c$. 
\begin{figure}[h]
    \centering
    \includegraphics[width=0.90\textwidth]{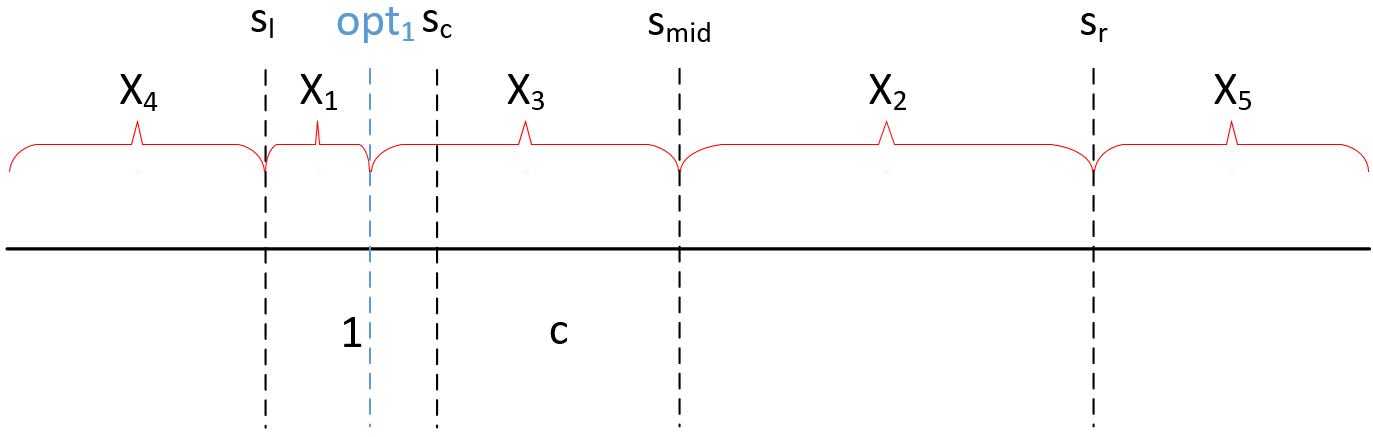}
    \caption{Case 3. $opt_1$ is between $s_{\ell}$ and $s_c$.}
    \label{Case 3}
\end{figure}

In this case, we divide $S_1$ into the same five parts defined in Case 2. (See Figure 4.)

Similar to Case 2, we have
\begin{align*}
	&\Delta_1 = \min \{d(X_1,s_{\ell}), d(X_3, s_r) + d(X_2, s_r)\},\\
	&BEST_1 = d(X_1,s_{\ell}) + d(X_2, s_r) + d(X_3, s_{\ell}) + d(X_4, s_{\ell}) + d(X_5, s_r).
\end{align*}

Let $\alpha$ be a constant parameter in $(0,1)$. We will determine the value of $\alpha$ later.

If $d(X_2, s_r) + d(X_3, s_{\ell}) \ge \alpha \cdot d(X_1,s_{\ell})$, we have
\begin{align*}
	BEST_1 \ge d(X_1,s_{\ell}) + d(X_2, s_r) + d(X_3, s_{\ell})
		   \ge (1+ \alpha) \cdot d(X_1, s_{\ell})
		   \ge (1+ \alpha) \cdot \Delta_1.
\end{align*}

This means $\Delta_1 \le 1/(1+\alpha) \cdot BEST_1$.

Otherwise, we have
\begin{align*}
	\alpha \cdot d(X_1, s_{\ell}) &> d(X_2, s_r) + d(X_3, s_{\ell}) \\
		   					 &\ge d(X_2, s_r) + |X_3| \cdot dist(opt_1, s_{\ell})\\
		 				     &\ge d(X_2, s_r) + \frac{|X_3|}{|X_1|} \cdot d(X_1, s_{\ell}).
\end{align*}

Now we use $OPT_1$ to bound $\Delta_1$. By definition, we have
\begin{align*}
	OPT_1 = d(S_1, opt_1) \ge d(X_2, opt_1) + d(X_5, opt_1).
\end{align*}

Note that $dist(v, opt_1) = dist(opt_1, s_r) - dist(v, s_r)$ holds for every $v \in X_2$ and $dist(v, opt_1) \ge dist(opt_1, s_r)$ holds for every $v \in X_5$. Hence, we have
\begin{align*}		  
	OPT_1 &\ge |X_2| \cdot dist(opt_1, s_r) - d(X_2, s_r) + |X_5| \cdot dist(opt_1, s_r)\\
		  &= (|X_2|+|X_5|) \cdot dist(opt_1, s_r) - d(X_2, s_r)\\
		  &\ge (|X_2|+|X_5|) \cdot \frac{2c+1}{2c+2} \cdot dist(s_{\ell}, s_r) - d(X_2, s_r).
\end{align*}

Note that $d(X_1, s_{\ell}) \le |X_1| \cdot dist(s_{\ell}, opt_1) \le |X_1| \cdot dist(s_{\ell},s_r)/(2c+2)$. Hence,
\begin{flalign*}	
	OPT_1 &\ge (2c+1) \cdot \frac{|X_2|+|X_5|}{|X_1|} \cdot d(X_1, s_{\ell}) - d(X_2, s_r)\\
		  &\ge (2c+1) \cdot \frac{|X_1|-|X_3|}{|X_1|} \cdot d(X_1, s_{\ell}) - d(X_2, s_r) \\
		  &= (2c+1) \cdot d(X_1, s_{\ell}) - ((2c+1) \cdot \frac{|X_3|}{|X_1|} \cdot d(X_1, s_{\ell})+ d(X_2, s_r))\\
		  &\ge (2c+1) \cdot d(X_1, s_{\ell}) - (2c+1)(\frac{|X_3|}{|X_1|} \cdot d(X_1, s_{\ell})+ d(X_2, s_r))\\
		  &\ge (2c+1) \cdot d(X_1, s_{\ell}) - (2c+1) \cdot \alpha \cdot d(X_1, s_{\ell})\\
		  &= (2c+1) \cdot (1 - \alpha) \cdot d(X_1, s_{\ell})\\
		  &\ge (2c+1) \cdot (1 - \alpha) \cdot \Delta_1.
\end{flalign*}

This means that we have
\begin{equation*}
	\Delta_1 \le \frac{1}{(2c+1) \cdot (1-\alpha)} OPT_1.
\end{equation*}

Now we summarize those three cases. Let $c=1$ and $\alpha = 1/(2c+1) = 1/3$. At least one of the following will happen:
\begin{align*}
	&\Delta_1 \le \max \{\frac{1}{2}, \frac{2c+1}{2c+2}, \frac{1}{1 + \alpha}\} \cdot BEST_1 = \frac{3}{4}BEST_1,~ {\rm or}\\
	&\Delta_1 \le \frac{1}{(2c+1) \cdot (1-\alpha)} OPT_1 = \frac{1}{2} OPT_1.
\end{align*}

It is easy to see that we can get the similar bounds for $\Delta_2$:
\begin{align*}
	\Delta_2 \le \frac{3}{4} BEST_2 ~ {\rm or} ~ \Delta_2 \le \frac{1}{2} OPT_2.
\end{align*}

Due to that there is no relationship between $OPT_1$ and $BEST_1$, we consider $\Delta_1$ and $\Delta_2$ together. Without loss of generality, we assume $BEST_1 \ge BEST_2$. We divide the discussion into 4 cases:

\textbf{Case 1:} $\Delta_1 \le 3/4 \cdot BEST_1$, $\Delta_2 \le 3/4 \cdot BEST_2$. In this case, we have
\begin{equation*}
\Delta 	= \Delta_1 + \Delta_2
		\le \frac{3}{4} (BEST_1 + BEST_2)
		= \frac{3}{4} BEST
		\le \frac{3}{4} OPT.
\end{equation*}

\textbf{Case 2:} $\Delta_1 \le 1/2 \cdot OPT_1$, $\Delta_2 \le 1/2 \cdot OPT_2$. In this case, we have
\begin{equation*}
\Delta 	= \Delta_1 + \Delta_2
		\le \frac{1}{2} (OPT_1 + OPT_2)
		= \frac{1}{2}OPT.
\end{equation*}

\textbf{Case 3:} $\Delta_1 \le 1/2 \cdot OPT_1$, $\Delta_2 \le 3/4 \cdot BEST_2$. In this case, we have
\begin{align*}
\Delta &= \Delta_1 + \Delta_2\\
	   &\le \frac{1}{2} OPT_1 + \frac{3}{4}BEST_2\\
	   &\le \frac{1}{2} OPT + \frac{3}{8}BEST\\
	   &\le \frac{7}{8}OPT.
\end{align*}

\textbf{Case 4:} $\Delta_1 \le 3/4 \cdot BEST_1$, $\Delta_2 \le 1/2 \cdot OPT_2$. In this case, we have
\begin{align*}
\Delta &= \Delta_1 + \Delta_2\\
	   &\le \frac{3}{4} BEST_1 + BEST_2\\
	   &\le \frac{3}{4} BEST_1 + \frac{1}{8}BEST_1 + \frac{7}{8}BEST_2\\
	   &= \frac{7}{8} BEST\\
	   &\le \frac{7}{8} OPT.
\end{align*}

It is easy to see that $\Delta \le 7/8 \cdot OPT$ holds for all these four cases. Therefore, we have $COST - OPT = 2 \cdot \Delta \le 2 \cdot 7/8 \cdot OPT = 7/4 \cdot OPT$, i.e. $COST \le 2.75 \cdot OPT$. 

\end{proof}
\subsection{A lower bound of the approximation ratio of Mechanism 1}

We provide a lower bound for the approximation ratio of Mechanism 1.

Let $N$ be a positive integer. Let $N$ agents be located at $x=0$ with preference $\{F_1\}$, $\lfloor (1+\sqrt{2})N \rfloor$ agents be located at $x=1$ with preference $\{F_1\}$ and almost infinite (far greater than $N$) agents be located at $x=\sqrt{2}$ with preference $\{F_2\}$. It is easy to see that $(s_{\ell}, s_r)$ will be $(0, \sqrt{2})$. Note that the minimum social cost is $N$ while the social cost calcuted by Mechanism 1 is $\min(\lfloor (1+\sqrt{2})N \rfloor, \sqrt{2} N + \lfloor (1+\sqrt{2})N \rfloor (\sqrt{2}-1))$. When $N$ increases to infinity, the ratio of social cost by Mechanism 1  to the minimum social cost approaches $1+\sqrt{2}$. Therefore, $1+\sqrt{2}$ is a lower bound of the approximation ratio of Mechanism 1.

We conjecture that the exact approximation ratio of Mechanism 1 is $1+\sqrt{2}$.

\section{Generalization}

We have tried to generalize our mechanism to $k$ facilities $(k \ge 3)$. Our generalized mechanism first ignores agents' preferences and assumes all facilities are acceptable for each agent. Then it only uses the location profile to calculate the optimal locations of all $k$ facilities, which we denote as $s_1, s_2, ..., s_k$. Finally, the mechanism enumerates $F_i$'s location being at $s_j$, $k^k$ cases in total, and chooses the one that minimizes the social cost.

However, unfortunately, this generalized mechanism is not strategyproof when there are more than two facilities. We give an example to show that this mechanism is not able to prevent agents from misreporting their preferences when there are three facilities.

\begin{Exmp}
	Assume that $l_1$ and $l_2$ are two positive numbers satisfying $2l_2 < l_1 < 3l_2$. Let
	\begin{itemize}
	\item two agents whose preferences are $\{F_2\}$ be located at $x=0$,
	\item one agent whose preference is $\{F_2\}$ be located at $x=l_1 - l_2$,
	\item one agent whose preference is $\{F_2\}$ be located at $x=l_1$,
	\item one agent whose preference is $\{F_2, F_3\}$ be located at $x=l_1 + l_2$,
	\item almost infinite agents whose preferences are $\{F_1\}$ be located at $x=0$,
	\item almost infinite agents whose preferences are $\{F_3\}$ be located at $x=2l_1 + l_2$.
	\end{itemize}
\end{Exmp}

\begin{figure}[h]
    \centering
    \includegraphics[width=0.90\textwidth]{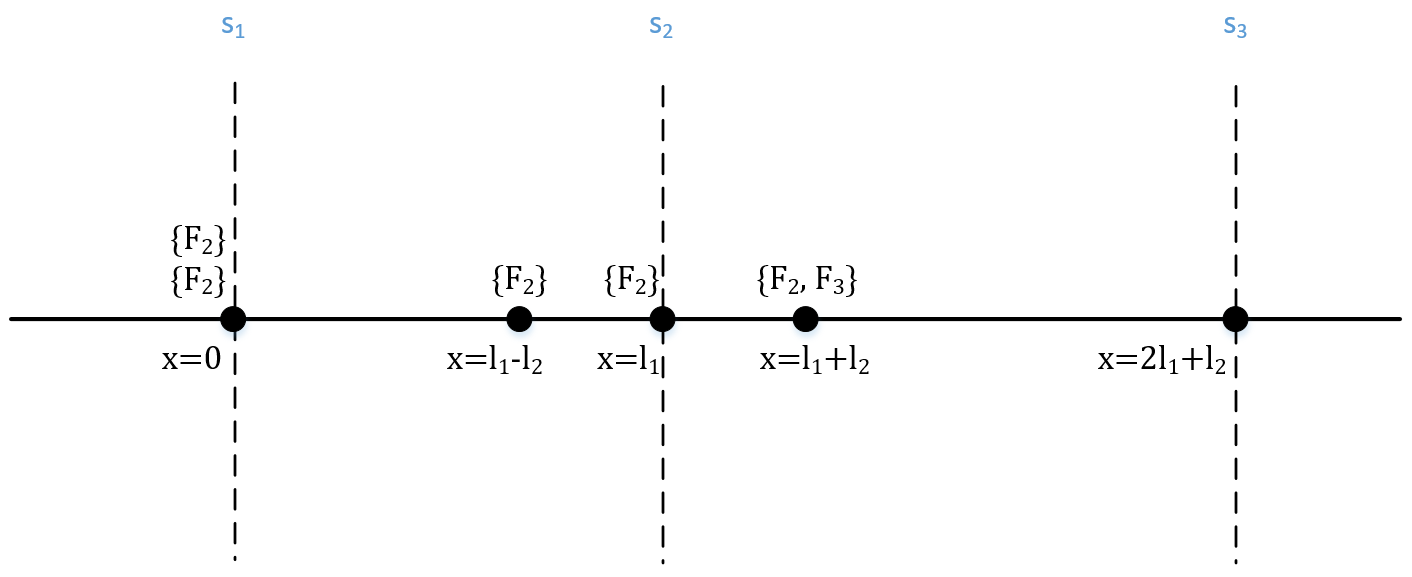}
    \caption{An example which shows the generalized mechanism is not strategyproof}
    \label{counter-example}
\end{figure}

It is easy to see that if we use the generalized mechanism on this input, $(s_1, s_2, s_3)$ will be $(0, l_1, 2 l_1 + l_2)$. (See Figure 5.) The almost infinite agents at $s_1$ with preferences $\{F_1\}$ force $F_1$ to be built at $s_1$. (Otherwise, the social cost will be almost infinite.) For the same reason, $F_3$ will be built at $s_3$.

If all agents report their true preferences, then $F_2$ will be built at $s_1$. However, if the agent at $l_1 + l_2$ (with preference $\{F_2, F_3\}$) misreports his preference as $\{F_2\}$, then $F_2$ will be built at $s_2$ and the cost of this agent will decrease from $l_1$ to $l_2$. Therefore, this generalized mechanism fails to prevent this agent from misreporting his preference.

In this example, we find that the social cost will increase if $F_2$ is still built at $s_1$ after the agent at $l_1 + l_2$ misreports his preference. However, this situation will not happen when there are only two facilities because the mechanism chooses the facilities' locations only between two locations $s_{\ell}$ and $s_r$. That is the reason why the mechanism is strategyproof when there are only two facilities but not strategyproof for three facilities.

Consider the facility location game with $k$ facilities, where $k>3$. For every $i \in \{4, 5, ..., k\}$, we add almost infinite agents whose locations are $x = (i-1) \cdot l_1 + l_2$ and preferences are $\{F_i\}$ to the example above. It is easy to see that the new example shows that the generalized mechanism is not strategyproof when there are $k$ facilities. Hence, this generalized mechanism is not strategyproof for the facility location game with more than two facilities.
\section{Conclusion}

To summarize, we have designed a strategyproof mechanism whose approximation ratio is no more than 2.75 and no less than $1+\sqrt{2}$. It is an open question whether the exact approximation ratio of this mechanism is $1+\sqrt{2}$. Furthermore, the bounds for approximation ratios of strategyproof mechanisms for two facilities do not match yet.

Besides, it is also an interesting open question to design a strategyproof mechanism with constant approximation ratio for the facility location game with three or more facilities.

\bibliographystyle{splncs04}
\bibliography{bib/Bibliography-File}

\end{document}